\documentclass[journal,onecolumn]{IEEEtran}
\usepackage{amssymb}
\usepackage{amsmath}
\usepackage{longtable}
\newtheorem{theorem}{Theorem}[section]
\newtheorem{lemma}[theorem]{Lemma}
\newtheorem{proposition}[theorem]{Proposition}
\newtheorem{corollary}[theorem]{Corollary}

\newenvironment{proof}[1][Proof]{\begin{trivlist}
\item[\hskip \labelsep {\bfseries #1}]}{\end{trivlist}}
\newenvironment{definition}[1][Definition]{\begin{trivlist}
\item[\hskip \labelsep {\bfseries #1}]}{\end{trivlist}}
\newenvironment{example}[1][Example]{\begin{trivlist}
\item[\hskip \labelsep {\bfseries #1}]}{\end{trivlist}}
\newenvironment{remark}[1][Remark]{\begin{trivlist}
\item[\hskip \labelsep {\bfseries #1}]}{\end{trivlist}}

\ifCLASSINFOpdf
\else
\fi
\begin{document}
%
\title{Special values of Kloosterman sums
and binomial bent functions}

\author{Chunming~Tang,
 Yanfeng~Qi
\thanks{C. Tang is with School of Mathematics and Information, China West Normal University, Sichuan Nanchong, 637002, China. e-mail: tangchunmingmath@163.com
}

\thanks{Y. Qi is with LMAM, School of Mathematical Sciences, Peking University, Beijing, 100871, and Aisino corporation Inc.,  Beijing, 100097,  China
}
}


\maketitle

\begin{abstract}
Let $p\ge 7$, $q=p^m$. $K_q(a)=\sum_{x\in \mathbb{F}_{p^m}} \zeta^{\mathrm{Tr}^m_1(x^{p^m-2}+ax)}$ is the Kloosterman sum of $a$ on $\mathbb{F}_{p^m}$,
where $\zeta=e^{\frac{2\pi\sqrt{-1}}{p}}$. The value $1-\frac{2}{\zeta+\zeta^{-1}}$
of $K_q(a)$ and its conjugate have close relationship with a class of binomial function
with Dillon exponent. This paper first presents some necessary conditions for $a$ such that
$K_q(a)=1-\frac{2}{\zeta+\zeta^{-1}}$. Further, we prove that if $p=11$, for any $a$,
$K_q(a)\neq 1-\frac{2}{\zeta+\zeta^{-1}}$. And for $p\ge 13$, if $a\in \mathbb{F}_{p^s}$ and $s=\mathrm{gcd}(2,m)$, $K_q(a)\neq 1-\frac{2}{\zeta+\zeta^{-1}}$.
In application, these results explains some class of binomial regular bent functions does not exits.
\end{abstract}

\begin{IEEEkeywords}
Regular bent function, Walsh transform, Kloosterman sums, $\pi$-adic expansion,
cyclotomic fields
\end{IEEEkeywords}

%
\IEEEpeerreviewmaketitle

\section{Introduction}

Let $q=p^m$, where $p$ be a prime and $m$ be a positive integer.
Let $\mathbb{F}_q$ be a finite field with $q$ elements.
Let $\mathrm{Tr}^m_1$ be the trace function from $\mathbb{F}_q$ to $\mathbb{F}_p$, i.e.,
$\mathrm{Tr}_{1}^{m}(x):= x+ x^{p}+
x^{p^{2}}+\cdots+ x^{p^{m-1}}$.  Let $\zeta=e^{\frac{2\pi\sqrt{-1}}{p}}$ be
a primitive
p-th root of unity. The Kloosterman sum of $a$ on
$\mathbb{F}_q$ is
\begin{align*}
K_q(a)=1+\sum_{x\in \mathbb{F}_q^*}\zeta^{\mathrm{Tr}^m_1(\frac{1}{x}+ax)}, a\in \mathbb{F}_q.
\end{align*}
Kloosterman sums are related to the construction of some Dillon type bent functions.

Let $n=2m$. When $p=2$, Leander \cite{Ld} proved that monomial function $\mathrm{Tr}^n_1(a x^{t(q-1)})$($a \in \mathbb{F}_{p^n}^*$, $\mathrm{gcd}(t,q+1)=1$) is bent if and only if $K_q(a^{q+1})=0$, i.e., $a^{q+1}$ is the zero point of Kloosterman sum $K_q$. Helleseth and Kholosha \cite{HK} generalized Leader's results for $p>2$ and found that the Kloosterman sum $K_q(\alpha)$ does not take the  value  zero
for $p>3$. Kononen et al. \cite{KRV} proved this fact.
When $p=2,3$, there exist many zero points of Kloosterman sums. And Moisio \cite{Ms} proved that any zero point of Kloosterman sums does not belong to a proper subfield of $\mathbb{F}_q$.

When $p\ge 3$, binomial function $\mathrm{Tr}^n_1(ax^{t(q-1)})+bx^{\frac{p^n-1}{2}}$($a\in \mathbb{F}_{p^n}^*$) is studied by Jia et al. \cite{JZH} and Zheng et al. \cite{ZYH}, where $b\in \mathbb{F}_{p}$, $\mathrm{gcd}(t,q+1)=1$. And it is bent if and only if $K_q(a)=1-\frac{2}{\zeta^b+\zeta^{-b}}$. Hence, for determining such bent functions, it is important to study  the value $1-\frac{2}{\zeta^b+\zeta^{-b}}$ of Kloosterman sums. Kononen \cite{Ms} presented a solution for $b=0$.

Divisibility results for Kloosterman sums  are vital and have many applications. On divisibility results of Kloosterman sums, many work can be found in \cite{CHZ,GF,GFMM,GFLGM,Ms24}.
Moloney \cite{Ml} analyzed divisibility results for $K_q(a)$ by $p$-adic  methods.

This paper will study the special value $1-\frac{2}{\zeta^b+\zeta^{-b}}$ of Kloosterman sums. By the $\pi$-adic expansions of $K_q(a)$ and $1-\frac{2}{\zeta^b+\zeta^{-b}}$, we
obtain some necessary conditions for $a$ satisfying $K_q(a)=1-\frac{2}{\zeta^b+\zeta^{-b}}$, where $\pi$ is a prime of local field $Q_p(\zeta)$ satisfying $\pi^{p-1}+p=0$ and $\zeta\equiv 1+\pi \mod \pi^2$. Further, we prove that if $p=11$, for any $a$,
$K_q(a)\neq 1-\frac{2}{\zeta^b+\zeta^{-b}}$. And for $p\ge 13$, if $a\in \mathbb{F}_{p^s}$ and $s=\mathrm{gcd}(2,m)$,
$K_q(a)\neq 1-\frac{2}{\zeta^b+\zeta^{-b}}$. Hence, these results explain that some class of binomial regular bent functions does not exist.

The rest of the paper is organized as follows. Section 2 introduces some background knowledge.
Section 3 presents our main results on special values of Kloosterman sums.
Section 4 gives some results on bent functions for application.
In Section 5, we make a conclusion.

\section{Preliminaries}
\subsection{Local fields and Gauss sums}

Throughout this paper, let $q=p^m$, where $p$ is a prime and $m$ is a positive integer. Let $\mathbb{F}_{q}$ be a finite field with $q$ elements.
Let $\mathbb{F}_{q}^*$ be the multiplicative group of $\mathbb{F}_{q}$. For any $a\in \mathbb{F}_{q}$, the trace function from $\mathbb{F}_{q}$ to $\mathbb{F}_{p}$ is defined by $\mathrm{Tr}^m_1(a):=a+a^{p}+\cdots+a^{p^{m-1}}$.

Let $Q_p$ be the finite field composed of all the p-adic numbers and
$Z_p$ be its integer ring. Let $\zeta=e^{\frac{2\sqrt{-1}\pi}{p}}$ be a primitive p-th root of unity. Let $\xi$ be a prmitive $(q-1)-$th root of unity. Then there exists the field extension tower $Q_p\subseteq Q_p(\zeta)\subseteq Q_p(\zeta,\xi)$.
Further, $Q_p(\zeta)$ is a totally ramified extension of degree  $p-1$  over $Q_p$, and $Q_p(\zeta,\xi)$  is a unramified extension of degree $m$ over
$Q_p(\zeta)$. Take a prime element $\pi$ in $Q_p(\zeta)$ and $Q_p(\zeta,\xi)$ such that
$$\pi^{p-1}+p=0, \zeta\equiv 1+\pi \mod \pi^2. $$

Let $\beta\equiv \xi \mod \pi$, then $\beta\in \mathbb{F}_{q}$ is a generator of $\mathbb{F}_{q}^*$.  The Teichm\"{u}ller character is defined by $$\omega(\beta^i)=\xi^i.$$

Let $\widehat{\mathbb{F}_q^*}$ be the group of multiplicative character of $\mathbb{F}_q$. Obviously, $\omega(\cdot)$ is a generator of $\widehat{\mathbb{F}_q^*}$. More results on local fields can be found in \cite{S}. Let
$$
\widehat{Tr}_{1}^{m}(a)=\sum_{i\in \{0,\ldots,m-1\}}{\omega^{p^i}(a)}.
$$
$\widehat{\mathrm{Tr}}_{1}^{m}$ is called the lifted trace. And by the modular
property of Teichm\"{u}ller character, $\widehat{\mathrm{Tr}}_{1}^{m}(a)
\equiv {\mathrm{Tr}}_{1}^{m}(a) \mod p$.

The Gauss sum of a character $w^{-j}\in \widehat{\mathbb{F}_q^*}$ is defined by
\begin{align*}
g(j)=-\sum_{x\in \mathbb{F}_q}w^{-j}(x)\zeta^{\mathrm{\mathrm{Tr}}^m_1(x)}.
\end{align*}

The following Stickelberger's theorem is helpful for the divisibility results of Gauss sums. And this theorem is a direct consequence of the Gross-Koblitz formula \cite{GK}.
\begin{theorem}\label{g(j)}
Let $1\le j<q-1$ and $j=j_0+j_1p+\cdots+j_{m-1}p^{m-1}$, then
\begin{align*}
g(j)\equiv \frac{1}{j_0!j_1!\cdots j_{m-1}!}\pi^{wt_p(j)} \mod \pi^{wt_p(j)+p-1},
\end{align*}
where $wt_p(j)=j_0+j_1+\cdots+j_{m-1}$.
\end{theorem}

\subsection{Kloosterman sums}
\begin{definition}
Let $a\in \mathbb{F}_q$. The Kloosterman sum of $a$ is
\begin{align}\label{Kq}
K_q(a)=\sum_{x\in \mathbb{F}_{q}}\zeta^{\mathrm{\mathrm{Tr}}^m_1(x^{q-2}+ax)}=1+\sum_{x\in \mathbb{F}_{q}^*}\zeta^{\mathrm{\mathrm{Tr}}^m_1(\frac{1}{x}+ax)}.
\end{align}
\end{definition}
Since $K_q(a)\in Q_p(\zeta)$, there is a unique $\pi$-adic expansion $K_q(a)=a_0+a_1\pi+\cdots+a_i \pi^i+\cdots$, where $a_i\in Q_p$ and $a_i^p=a_i$.

From the identity \cite{W}
\begin{align}\label{Kmorq}
K_q(a)\equiv \sum_{j=1}^{q-2} (g(j))^2 \omega^j(a) \mod q,
\end{align}
and Theorem \ref{g(j)}, Moloney\cite{Ml} studied the divisibility results of $K_q(a)$. In particular, Moloney considered the case for $p=2,3$ and computed coefficients $a_0,a_1,\cdots,a_7$ of the $\pi$-adic expansion of $K_q(a)$ for $p\geq 7$.

\subsection{Binomial regular bent functions}

The Walsh transform of a $p$-ary function $f(x):\mathbb{F}_{p^n}\rightarrow \mathbb{F}_p$ is defined by
\begin{align*}
W_f(\lambda)=\sum_{x\in \mathbb{F}_{p^n}}e^{\frac{2\pi \sqrt{-1}}{p}(f(x)-\mathrm{\mathrm{Tr}}^m_1(\lambda x))}.
\end{align*}
$f(x)$ is a $p$-ary bent function, if for any $\lambda \in \mathbb{F}_{p^n}$, $|W_f(\lambda)|=p^{\frac{n}{2}}$. Further, $f$ is regular if $W_f(\lambda)=p^{\frac{n}{2}}e^{\frac{2\pi\sqrt{-1}}{p}f^*(\lambda)}$, where $f^*$ is some  $p$-ary
function from $\mathbb{F}_{p^n}$ to $\mathbb{F}_p$.

If $n=2m$, Jia et al. \cite{JZH} considered the binomial $p$-ary function of the form
\begin{align*}
f_{a,b,t}=\mathrm{\mathrm{Tr}}^n_1(a x^{t(p^m-1)})+bx^{\frac{p^n-1}{2}}
\end{align*}
and obtained the following theorem.
\begin{theorem}
Let $n=2m$, $q=p^m\equiv 3 \mod 4$ or $p=3$.  Let $\mathrm{gcd}(t,p^m+1)=1$, $a \in \mathbb{F}_{q^2}$, and $b\in \mathbb{F}_p$. Then
$f_{a,b,t}$ is a regular bent function if and only if $K_q(a^{q+1})=1-\frac{2}{\zeta^b+\zeta^{-b}}.$
\end{theorem}

Zheng et al. \cite{ZYH} improved results in \cite{JZH} and gave the following results.
\begin{theorem}\label{Zh}
Let $n=2m$, $q=p^m$, $\mathrm{gcd}(t,p^m+1)=1$, $a \in \mathbb{F}_{q^2}$, and $b\in \mathbb{F}_p$. Then
$f_{a,b,t}$ is regular bent if and only if $K_q(a^{q+1})=1-\frac{2}{\zeta^b+\zeta^{-b}}.$
\end{theorem}
\begin{theorem}
Let $n=2m$, $q=p^m\equiv 1\mod 4$, $t\equiv 2 \mod 4$, $\mathrm{gcd}(\frac{t}{2},p^m+1)=1$, $a \in \mathbb{F}_{q^2}$, and $b\in \mathbb{F}_p$. Then
$f_{a,b,t}$ is regular bent if and only if $K_q(a^{q+1})=1-\frac{2}{\zeta^b+\zeta^{-b}}.$
\end{theorem}

\section{Kloosterman sums}
For abbreviation, we denote $\sum_{i_1,\cdots,i_t\in \{0,\cdots,m-1\},i_j\neq i_k(j\neq k)}A_{i_1,\cdots,i_t}$ by $\sum A_{i_1,\cdots,i_t}$.

\subsection{The $\pi$-adic expansion of Kloosterman sums}
\begin{lemma}\label{Fmsum}
Let $m\ge 1$ and $a\in \mathbb{F}_{p^m}$, then

\rm{(1)} $\sum \omega^{p^i+p^j}(a)=(\widehat{\mathrm{Tr}}^m_1(a))^2-\widehat{\mathrm{Tr}}^m_1(a^2);$

\rm{(2)} $\sum \omega^{2p^i+p^j}(a)=(\widehat{\mathrm{Tr}}^m_1(a))(\widehat{\mathrm{Tr}}^m_1(a^2))-\widehat{\mathrm{Tr}}(a^3);$

\rm{(3)} $\sum \omega^{p^i+p^j+p^k}(a)=(\widehat{\mathrm{Tr}}^m_1(a))^3-3(\widehat{\mathrm{Tr}}^m_1(a))
(\widehat{\mathrm{Tr}}^m_1(a^2))+2\widehat{\mathrm{Tr}}^m_1(a^3)$;

\rm{(4)} $\sum \omega^{2p^i+2p^j}(a)=(\widehat{\mathrm{Tr}}^m_1(a^2))^2-\widehat{\mathrm{Tr}}^m_1(a^4);$

\rm{(5)} $\sum \omega^{3p^i+p^j}(a)=(\widehat{\mathrm{Tr}}^m_1(a))(\widehat{\mathrm{Tr}}^m_1(a^3))-\widehat{\mathrm{Tr}}^m_1(a^4);$

\rm{(6)} $\sum \omega^{2p^i+p^j+p^k}(a)=(\widehat{\mathrm{Tr}}^m_1(a))^2(\widehat{\mathrm{Tr}}^m_1(a^2))-
2(\widehat{\mathrm{Tr}}^m_1(a))(\widehat{\mathrm{Tr}}^m_1(a^3))-(\widehat{\mathrm{Tr}}^m_1(a^2))^2+
2\widehat{\mathrm{Tr}}^m_1(a^4);$

\rm{(7)} $\sum \omega^{p^i+p^j+p^k+p^l}(a)=(\widehat{\mathrm{Tr}}^m_1(a))^4-
6(\widehat{\mathrm{Tr}}^m_1(a))^2(\widehat{\mathrm{Tr}}^m_1(a^2))+3(\widehat{\mathrm{Tr}}^m_1(a^2))^2+
8(\widehat{\mathrm{Tr}}^m_1(a))(\widehat{\mathrm{Tr}}^m_1(a^3))-6\widehat{\mathrm{Tr}}^m_1(a^4).$
\end{lemma}
\begin{proof}
(1) From
\begin{align*}
(\sum \omega^{p^i}(a))^2=\sum \omega^{p^i+p^j}(a)+\sum \omega^{2p^i}(a),
\end{align*}
this result can be obviously obtained.

(2) From
\begin{align*}
(\sum \omega^{2p^i}(a))(\sum \omega^{p^j}(a))=\sum \omega^{2p^i+p^j}(a)+\sum \omega^{3p^i}(a),
\end{align*}
this result can be obviously obtained.

(3) From
\begin{align*}
(\sum \omega^{p^i}(a))(\sum \omega^{p^j+p^k}(a))=\sum \omega^{p^i+p^j+p^k}(a)+2\sum \omega^{2p^i+p^j}(a),
\end{align*}
Result (1) and Result (2), this result can be obviously obtained.

(4) From Result (1), this result can be obviously obtained.

(5) From
\begin{align*}
(\sum \omega^{3p^i}(a))(\sum \omega^{p^j}(a))=\sum \omega^{3p^i+p^j}(a)+\sum \omega^{4p^i}(a),
\end{align*}
this result can be obviously obtained.

(6) From
\begin{align*}
(\sum \omega^{2p^i+p^j}(a))(\sum \omega^{p^k}(a))=\sum \omega^{2p^i+p^j+p^k}(a)+\sum \omega^{3p^i+p^j}(a)+\sum \omega^{2p^i+2p^j}(a),
\end{align*}
Result (2), Result (4), and Result (5), this result can be obviously obtained.

(7) From
\begin{align*}
(\sum \omega^{p^i}(a))(\sum \omega^{p^j+p^k+p^l}(a))=\sum \omega^{p^i+p^j+p^k+p^l}(a)+3\sum \omega^{2p^j+p^k+p^l}(a),
\end{align*}
Result (3), and Result (6), this result can be obviously obtained.
\end{proof}

The following proposition is a generalization of results in Chapter 5.3 in \cite{Ml} of Moloney.
\begin{proposition}\label{a8}
Let $p$ be a prime greater than 11, $m\ge 1$, and $a\in \mathbb{F}_{p^m}^*$. $K_q(a)$ is the Kloosterman sum of $a$ defined in (\ref{Kq}). Let the $\pi$-adic expansion of
$K_q(a)$ in $Q_p(\zeta)$ be of the form $K_q(a)=\sum_{i=0}^{+\infty}a_i \pi ^i,$
where $\zeta$ is a primitive $p-$th root of unity, $\pi$ is a prime of $Q_p(\zeta)$ satisfying $\pi ^{p-1}=-p$, $\zeta\equiv 1+\pi \mod \pi ^2$, and $a_i^p=a_i$. Then

\rm{(1)} [Chapter 5.3 in \cite{Ml}] $a_0,a_2,a_4,a_6$ are determined by
\begin{align*}
a_0&=0,\\
a_2&\equiv -\mathrm{\mathrm{Tr}}^m_1(a) \mod p,\\
a_4&\equiv  \frac{1}{4}(\mathrm{\mathrm{Tr}}^m_1(a^2)-2(\mathrm{\mathrm{Tr}}^m_1(a))^2)\mod p,\\
a_6&\equiv -\frac{1}{36}(4\mathrm{\mathrm{Tr}}^m_1(a^3)+6(\mathrm{\mathrm{Tr}}^m_1(a))^3-9(\mathrm{\mathrm{Tr}}^m_1(a))(\mathrm{\mathrm{Tr}}^m_1(a^2)))\mod p;
\end{align*}

\rm{(2)} $a_8\equiv -\frac{1}{576}(24(\mathrm{\mathrm{Tr}}^m_1(a))^4-72(\mathrm{\mathrm{Tr}}^m_1(a))^2(\mathrm{\mathrm{Tr}}^m_1(a^2))+64(\mathrm{\mathrm{Tr}}^m_1(a))(\mathrm{\mathrm{Tr}}^m_1(a^3))+18(\mathrm{\mathrm{Tr}}^m_1(a^2))^2-33\mathrm{\mathrm{Tr}}^m_1(a^4)) \mod p$.

\rm{(3)} $a_{2i+1}=0$, $i=0,1,2,\cdots.$
\end{proposition}
\begin{proof}
(1) This can be found in Chapter 5.3 in \cite{Ml} by Moloney.

(2) From Result (1), we have
\begin{align*}
-\sum_{wt_p(j)=1}g(j)^2\omega^j(a)&\equiv -\mathrm{\mathrm{Tr}}^m_1(a)\pi^2 \mod \pi^{p+1},\\
-\sum_{wt_p(j)=2}g(j)^2\omega^j(a)&\equiv \frac{1}{4}(\mathrm{\mathrm{Tr}}^m_1(a^2)-2(\mathrm{\mathrm{Tr}}^m_1(a))^2)\pi^4 \mod \pi^{p+3},\\
-\sum_{wt_p(j)=3}g(j)^2\omega^j(a)&\equiv -\frac{1}{36}(4\mathrm{\mathrm{Tr}}^m_1(a^3)+6(\mathrm{\mathrm{Tr}}^m_1(a))^3-
9(\mathrm{\mathrm{Tr}}^m_1(a))(\mathrm{\mathrm{Tr}}^m_1(a^2)))\pi^6 \mod \pi^{p+5},
\end{align*}
From $p\ge 11$ and (\ref{Kmorq}), we have
\begin{align*}
K_q(a)&\equiv -\mathrm{\mathrm{Tr}}^m_1(a)\pi^2+\frac{1}{4}(\mathrm{\mathrm{Tr}}^m_1(a^2)-2(\mathrm{\mathrm{Tr}}^m_1(a))^2)\pi^4\\
&-\frac{1}{36}(4\mathrm{\mathrm{Tr}}^m_1(a^3)+6(\mathrm{\mathrm{Tr}}^m_1(a))^3-9(\mathrm{\mathrm{Tr}}^m_1(a))(\mathrm{\mathrm{Tr}}^m_1(a^2)))\pi^6\\
&-\sum_{wt_p(j)=4}g(j)^2\omega^j(a)\mod \pi^{10}.
\end{align*}
From Theorem \ref{g(j)}, we have
\begin{align*}
-\sum_{wt_p(j)=4}g(j)^2\omega^j(a)\equiv &-\frac{\pi^8}{576}(\sum \omega^{4p^i}(a)+36\sum_{i< j} \omega^{2p^i+2p^j}(a)+16\sum \omega ^{3p^i+p^j}(a)\\
&+144\sum_{j< k} \omega^{2p^i+p^j+p^k}(a)+576\sum_{i< j < k< l} \omega^{p^i+p^j+p^j+p^k}(a)) \\
\equiv &-\frac{\pi^8}{576}(\sum \omega^{4p^i}(a)+18\sum \omega^{2p^i+2p^j}(a)+16\sum \omega ^{3p^i+p^j}(a)\\
&+72\sum \omega^{2p^i+p^j+p^k}(a)+24\sum \omega^{p^i+p^j+p^j+p^k}(a)) \mod \pi^{10}.
\end{align*}
From Result (4), (5), (6), (7) in Lemma \ref{Fmsum}, we have
\begin{align*}
a_8\equiv &-\frac{1}{576}(24(\widehat{\mathrm{Tr}}^m_1(a))^4-72(\widehat{\mathrm{Tr}}^m_1(a))^2
(\widehat{\mathrm{Tr}}^m_1(a^2))+64(\widehat{\mathrm{Tr}}^m_1(a))(\widehat{\mathrm{Tr}}^m_1(a^3))\\
&+18(\widehat{\mathrm{Tr}}^m_1(a^2))^2-33\widehat{\mathrm{Tr}}^m_1(a^4)) \mod \pi.
\end{align*}
Note that $\widehat{\mathrm{Tr}}^m_1(a)\equiv \mathrm{\mathrm{Tr}}^m_1(a)\mod p$, $\widehat{\mathrm{Tr}}^m_1(a^2)\equiv \mathrm{\mathrm{Tr}}^m_1(a^2)\mod p$, $\widehat{\mathrm{Tr}}^m_1(a^3)\equiv \mathrm{\mathrm{Tr}}^m_1(a^3)\mod p$, and $\widehat{\mathrm{Tr}}^m_1(a^4)\equiv \mathrm{\mathrm{Tr}}^m_1(a^4)\mod p$. Then
\begin{align*}
a_8\equiv &-\frac{1}{576}(24(\mathrm{\mathrm{Tr}}^m_1(a))^4-72(\mathrm{\mathrm{Tr}}^m_1(a))^2
(\mathrm{\mathrm{Tr}}^m_1(a^2))+64(\mathrm{\mathrm{Tr}}^m_1(a))(\mathrm{\mathrm{Tr}}^m_1(a^3))\\
&+18(\mathrm{\mathrm{Tr}}^m_1(a^2))^2-33\mathrm{\mathrm{Tr}}^m_1(a^4)) \mod \pi.
\end{align*}
Note that $a_8\in Z_p$, hence
\begin{align*}
a_8\equiv&-\frac{1}{576}(24(\mathrm{\mathrm{Tr}}^m_1(a))^4-72(\mathrm{\mathrm{Tr}}^m_1(a))^2(\mathrm{\mathrm{Tr}}^m_1(a^2))+64(\mathrm{\mathrm{Tr}}^m_1(a))(\mathrm{\mathrm{Tr}}^m_1(a^3))\\
&+18(\mathrm{\mathrm{Tr}}^m_1(a^2))^2-33\mathrm{\mathrm{Tr}}^m_1(a^4)) \mod p.
\end{align*}

(3) Note that $\sigma_{-1}$ is a Galois automorphism of  $Q_p(\zeta)$ satisfying $\sigma_{-1}(\zeta)=\zeta^{-1}$. The action   on Teichm\"{u}ller  elements in $Q_p$ by $\sigma_{-1}$ is ordinary. For any positive integer $k$, we have
\begin{align}\label{sigma-1}
\sigma_{-1}(K_q(a))\equiv a_0+a_1(\sigma_{-1}(\pi))+a_2(\sigma_{-1}(\pi))^2+\cdots +a_{2k+1}(\sigma_{-1}(\pi))^{2k+1} \mod \pi^{2k+2},
\end{align}
Since $\pi^{p-1}=-p$, $(\sigma_{-1}(\pi))^{p-1}=-p$ and $\sigma_{-1}(\pi)=w_{p-1}\pi$,
where $w_{p-1}$ is a $(p-1)$-th root of unity in $Q_p$. Note that
\begin{align}\label{zeta-pi}
\zeta\equiv 1+\pi \mod \pi^2.
\end{align}
Further,
\begin{align*}
\zeta^{-1}\equiv 1-\pi \mod \pi^2.
\end{align*}
From the action on (\ref{zeta-pi}) by $\sigma_{-1}$,
\begin{align*}
\zeta^{-1}\equiv 1+w_{p-1}\pi \mod \pi^2.
\end{align*}
Then
\begin{align*}
w_{p-1}\equiv -1 \mod \pi.
\end{align*}
Hence, $w_{p-1}=-1$ and $\sigma_{-1}(\pi)=-\pi$.
From (\ref{sigma-1}),
\begin{align*}
\sigma_{-1}(K_q(a))\equiv a_0-a_1\pi+a_2\pi^2-a_3\pi^3+\cdots+a_{2k}\pi^{2k}-a_{2k+1}\pi^{2k+1} \mod \pi^{2k+2}.
\end{align*}
From the definition of $K_q(a)$, $\sigma_{-1}(K_q(a))=K_q(a)$. Then
\begin{align*}
&a_0-a_1\pi+a_2\pi^2-a_3\pi^3+\cdots+a_{2k}\pi^{2k}-a_{2k+1}\pi^{2k+1}\\
\equiv &a_0-a_1\pi+a_2\pi^2-a_3\pi^3+\cdots+a_{2k}\pi^{2k}-a_{2k+1}\pi^{2k+1} \mod \pi^{2k+2}.
\end{align*}
Hence, $a_1=a_3=\cdots=\cdots=a_{2k+1}=0$. From the random choice of $k$, Result (3) holds.
\end{proof}
\begin{example}
Let $p=11$, $q=p^4$, $\zeta^p=1$, $\pi^{10}+11=0$, $\zeta\equiv 1+\pi \mod \pi^2$, $GF(q)=GF(p)(\beta)$, $\beta^4+8\beta^2+10\beta+2=0$, and $a=\beta^{2092}$.
From direct computation, we have
$K_q(a)=-2\pi^8 + 5\pi^6 + 4\pi^4 + 4\pi^2 \mod \pi^{10}$. From $\mathrm{\mathrm{Tr}}^4_1(a)=7$, $\mathrm{\mathrm{Tr}}^4_1(a^2)=4$, $\mathrm{\mathrm{Tr}}^4_1(a^3)=4$, and $\mathrm{\mathrm{Tr}}^4_1(a^4)=8$, the expansion of $K_q(a)$ is just
the result in Proposition \ref{a8}.
\end{example}
\begin{remark}
Result (1) in Proposition \ref{a8} still holds for $p=7$. Result (3) in Proposition \ref{a8} can be generalized, i.e, if $x\in Z_p(\zeta+\zeta^{-1})$,
\begin{align*}
x=a_0+a_2\pi^2+a_4\pi^4+a_6\pi^{6}+\cdots.
\end{align*}
The proof is similar.
\end{remark}

From the above discussion, we have $\sigma_{-1}(\zeta)\equiv 1-\pi \mod \pi$. Actually, we have the following general result.
\begin{lemma}\label{sigma}
Let $i$ be a nonzero integer satisfying $-\frac{p-1}{2}\le i\le \frac{p-1}{2}$, $\sigma_{i}$ is be a Galois automorphism of $Q_p(\zeta,\xi)$ such that $\sigma_{i}(\zeta)=\zeta^i$. Let $\alpha\in Q_p(\zeta,\xi)$ and $\nu_{\pi}(\alpha)\ge 0$. Let the $\pi$-adic expansion of $\alpha$ be
\begin{align*}
\alpha=a_0+a_1\pi+a_2\pi^2+\cdots,
\end{align*}
where $a_i\in Q_p(\zeta,\xi)$ and $a_i^{q}=a_i$. Then
\begin{align*}
\sigma_{i}(\alpha)=a_0+a_1\omega(i)\pi+\cdots+a_j\omega^j(i)\pi^j+\cdots,
\end{align*}
where $\omega(\cdot)$ is the Teichm\"{u}ller character. In particular, $\sigma_{i}(\alpha)\equiv a_0+a_1i\pi\mod \pi^{2}$.
\end{lemma}
\begin{proof}
Note that $\sigma_{i}$ acts  ordinary on  unramified extension $Q_p(\xi)$. If $a^q=a$, $\sigma_i(a)=a$.

Since $\zeta\equiv 1+\pi\mod \pi^2$, $\sigma_i(\zeta)\equiv 1+\sigma_i(\pi)\mod \pi^2$.
On the other hand, $\sigma_i(\zeta)\equiv (1+\pi)^i=1+i\pi\mod \pi^2$. Then $\sigma_i(\pi)\equiv i\pi\mod \pi^2$.
From the definition of $\pi$, $\sigma_i(\pi)=w\pi$, where $w$ is some $(p-1)$-th
root of unity. From $w\equiv i\mod \pi$,
$w=\omega(i)$, $\sigma_i(\pi)=\omega(i)\pi$. Then we have
\begin{align*}
\sigma_{i}(\alpha)=a_0+a_1\omega(i)\pi+\cdots+a_j\omega^j(i)\pi^j+\cdots.
\end{align*}
From $\omega(i)=i\mod p$, we have $\sigma_{i}(\alpha)\equiv a_0+a_1i\pi\mod \pi^{2}$.
Hence, this lemma holds.
\end{proof}

\subsection{The $\pi$-adic expansion of elements in $Q_{p}(\zeta)$}
\begin{proposition}\label{zetapi}
Let $p$ be a prime greater than $11$ and  $\zeta$ be
a $p$-th root of unity satisfying $\zeta\equiv 1+\pi\mod \pi^2$. Then

\rm{(1)} $\zeta\equiv 1 + \pi + \frac{1}{2}\pi^2 + \frac{1}{6}\pi^3 + \frac{1}{24}\pi^4 + \frac{1}{120}\pi^5 + \frac{1}{720}\pi^6 + \frac{1}{5040}\pi^7 + \frac{1}{40320}\pi^8  \mod \pi^{9}$;

\rm{(2)} $\zeta^{-1}\equiv 1 - \pi + \frac{1}{2}\pi^2 - \frac{1}{6}\pi^3 + \frac{1}{24}\pi^4 - \frac{1}{120}\pi^5 + \frac{1}{720}\pi^6 - \frac{1}{5040}\pi^7 + \frac{1}{40320}\pi^8  \mod \pi^{9}$;

\rm{(3)} $\zeta+\zeta^{-1}\equiv 2 +\pi^2 + \frac{1}{12}\pi^4 + \frac{1}{360}\pi^6 + \frac{1}{20160}\pi^8  \mod \pi^{10}$;

\rm{(4)} $1-\frac{2}{\zeta+\zeta^{-1}}\equiv \frac{1}{2}\pi^2 - \frac{5}{24}\pi^4 + \frac{61}{720}\pi^6 - \frac{277}{8064}\pi^8  \mod \pi^{10}$.
\end{proposition}
\begin{proof}
(1) Let the $\pi$-adic expansion of $\zeta$ be $\zeta=1+\pi+a_2\pi^2+a_3\pi^3+\cdots$, where $a_i^p=a_i$ and $a_i\in Z_p$.
For simplicity, let $a_0=a_1=1$. Then
\begin{align*}
(1+\pi+a_2\pi^2+\cdots)^p-1=&\pi^p(1+a_2\pi+a_3\pi^2+\cdots)^p
+\sum_{i=1}^{p-1}\binom{p}{i}\pi^i(1+a_2\pi+a_3\pi^2+\cdots)^i\\
\equiv &\pi^p-\sum_{i=1}^{8}\pi^{p-1+i}\frac{\binom{p}{i}}{p}(1+a_2\pi+a_3\pi^2+\cdots+a_{9-i}\pi^{8-i})^i   \mod \pi^{p+8},\\
\equiv &-\pi^p(a_2\pi+a_3\pi^2+\cdots+a_8\pi^7)\\
&-\sum_{i=2}^{8}\pi^{p-1+i}\frac{\binom{p}{i}}{p}(1+a_2\pi+a_3\pi^2+\cdots+a_{9-i}\pi^{8-i})^i   \mod \pi^{p+8},\\
\equiv &-\pi^p((a_2\pi+a_3\pi^2+\cdots+a_8\pi^7)\\
&+\sum_{i=2}^{8}\pi^{i-1}\frac{\binom{p}{i}}{p}(1+a_2\pi+a_3\pi^2+\cdots+a_{9-i}\pi^{8-i})^i) \mod \pi^{p+8},
\end{align*}
From $(1+\pi+a_2\pi^2+\cdots)^p-1=0$, we have
\begin{align*}
(a_2\pi+a_3\pi^2+\cdots+a_8\pi^7)+\sum_{i=2}^{8}\pi^{i-1}
\frac{\binom{p}{i}}{p}(1+a_2\pi+a_3\pi^2+\cdots+a_{9-i}\pi^{8-i})^i\equiv 0 \mod \pi^{8},
\end{align*}
Denote $s_1\equiv a_2\pi+a_3\pi^2+\cdots+a_8\pi^7 \mod \pi^8$ and $s_i\equiv \pi^{i-1}\frac{\binom{p}{i}}{p}(1+a_2\pi+a_3\pi^2+\cdots+a_{9-i}\pi^{8-i})^i\mod \pi^8$$(i=2,3\cdots,8)$. From direct computation, we have
\begin{align*}
s_2\equiv &- (a_2a_6 +a_3a_5 +\frac{1}{2}a_4^2 +a_7)\pi^7 - (a_2a_5 +a_3a_4 +a_6)\pi^6 \\
&- (a_2a_4 + \frac{1}{2}a_3^2 + a_5)\pi^5 - (a_2a_3 + a_4)\pi^4 - (\frac{1}{2}a_2^2 + a_3)\pi^3 - a_2\pi^2 - \frac{1}{2}\pi \mod \pi^8.
\end{align*}
Accordingly, $s_3,\cdots,s_8$ can be computed. Then we have
\begin{align*}
\sum_{i=1}^{8}s_i\equiv \sum_{i=1}^{7}c_i\pi^i \mod \pi^8,
\end{align*}
where
\begin{align*}
c_1=&a_2 - 1/2,\\
c_2=&-a_2 + a_3 + 1/3,\\
c_3=& - 1/2a_2^2 +a_2 - a_3 + a_4 - 1/4,\\
c_4=& a_2^2 - a_2a_3 - a_2 + a_3 - a_4 + a_5 + 1/5,\\
c_5=&1/3a_2^3 - 3/2a_2^2 + 2a_2a_3 - a_2a_4 + a_2 -
    1/2a_3^2 - a_3 + a_4 - a_5 + a_6 - 1/6,\\
c_6=&- a_2^3 +a_2^2a_3 + 2a_2^2 - 3a_2a_3 + 2a_2a_4 - a_2a_5 - a_2 + a_3^2 - a_3a_4 + a_3 -a_4 + a_5 - a_6 + a_7 + 1/7,\\
c_7=& - 1/4a_2^4 + 2a_2^3 - 3a_2^2 a_3 + a_2^2 a_4 - 5/2a_2^2 + a_2 a_3^2 + 4a_2 a_3 - 3a_2 a_4 + 2a_2 a_5 - a_2 a_6 \\
&+ a_2 - 3/2 a_3^2 + 2a_3 a_4 - a_3 a_5 - a_3 - 1/2a_4^2 + a_4 - a_5 + a_6 - a_7 + a_8 - 1/8,
\end{align*}
Since $\sum_{i=1}^{7}c_i\pi^i\equiv 0 \mod \pi^8$, we have
\begin{align*}
 &a_2\equiv 1/2\mod p,~a_3\equiv 1/6 \mod p,~a_4\equiv 1/24\mod p,~a_5\equiv 1/120\mod p,\\
 &a_6\equiv 1/720\mod p,~a_7\equiv 1/5040\mod p,~ a_8\equiv 1/40320\mod p.
\end{align*}
Hence Result (1) holds.

(2) From $\zeta^{-1}\zeta=1$ and Result (1), Result (2) can be obviously obtained.

(3) From Result (1) and (2), Proposition \ref{a8}, this result can be obviously obtained.

(4) From Result (3), this result can be obviously obtained.
\end{proof}

\begin{example}
Let $p=37$, $\pi^{36}+37=0$, $\zeta^{p}=1$, and $\zeta\equiv 1+\pi \mod \pi^2$.
From direct computation, $\zeta\equiv 11\pi^8 + 14\pi^7 - 13\pi^6 - 4\pi^5 + 17\pi^4 - 6\pi^3 - 18\pi^2 + \pi + 1 \mod \pi^9$. From Proposition \ref{zetapi},
$\zeta\equiv 1 + \pi + \frac{1}{2}\pi^2 + \frac{1}{6}\pi^3 + \frac{1}{24}\pi^4 + \frac{1}{120}\pi^5 + \frac{1}{720}\pi^6 + \frac{1}{5040}\pi^7 + \frac{1}{40320}\pi^8  \mod \pi^{9}$. Note that $\frac{1}{2}\equiv -18 \mod 37$, $\frac{1}{6}\equiv -6 \mod 37$,$\frac{1}{24}\equiv 17\mod 37$,
$\frac{1}{120}\equiv -4\mod 37$, $\frac{1}{720}\equiv -13\mod 37$, and  $\frac{1}{5040}\equiv 14\mod 37$,$\frac{1}{40320}\equiv 11\mod 37$. The computation result of Proposition \ref{zetapi}
 is just the same as the direct computation.
\end{example}
\begin{corollary}
Let $p$ be a prime greater than $11$, and $\zeta$ be a primitive
$p$-th root of unity. Then $\prod_{i=1}^{\frac{p-1}{2}}(1-\frac{2}{\zeta^i+\zeta^{-i}})\equiv (\frac{-2}{p})p\mod p^2$.
\end{corollary}
\begin{proof}
Note that $1-\frac{2}{\zeta^i+\zeta^{-i}}=\sigma_i(1-\frac{2}{\zeta+\zeta^{-1}})$.
From Proposition \ref{zetapi} and Lemma \ref{sigma}, we have
$1-\frac{2}{\zeta^i+\zeta^{-i}}\equiv \frac{1}{2}i^2\pi^2\mod \pi^4$. Then
\begin{align*}
\prod_{i=1}^{\frac{p-1}{2}}(1-\frac{2}{\zeta^i+\zeta^{-i}})=
\prod_{i=1}^{\frac{p-1}{2}}(\frac{1}{2}i^2\pi^2)=\frac{1}{2^{\frac{p-1}{2}}}
\pi^{p-1}\prod_{i=1}^{\frac{p-1}{2}}i^2
=(\frac{2}{p})(-(\frac{-1}{p}))\pi^{p-1}
\equiv (\frac{-2}{p})p \mod \pi^{p+1}.
\end{align*}
Note that $\prod_{i=1}^{\frac{p-1}{2}}(1-\frac{2}{\zeta^i+\zeta^{-i}})\in Z_p$. Then $\prod_{i=1}^{\frac{p-1}{2}}(1-\frac{2}{\zeta^i+\zeta^{-i}})\equiv (\frac{-2}{p})p \mod p^2$.
\end{proof}
\subsection{Special values of Kloosterman sums}

\begin{proposition}\label{ai^2}
Let $i$ be an integer such that $1\le i \le \frac{p-1}{2}$, and $a\in \mathbb{F}_{q}$, then $K_q(a)=1-\frac{2}{\zeta^i+\zeta^{-i}}$ if and only if $K_q(\frac{1}{i^2}a)=1-\frac{2}{\zeta+\zeta^{-1}}$.
\end{proposition}

\begin{proof}
From the definition of Kloosterman sums, we have
\begin{align*}
K_q(\frac{1}{i^2}a)=&\sum_{x\in \mathbb{F}_{q}}\zeta^{\mathrm{\mathrm{Tr}}^m_1(\frac{1}{i^2}ax+x^{q-2})}
=\sum_{x\in \mathbb{F}_{q}}\zeta^{\mathrm{\mathrm{Tr}}^m_1(\frac{1}{i^2}a(ix)+(ix)^{q-2})}
=\sum_{x\in \mathbb{F}_{q}}\zeta^{\frac{1}{i}\mathrm{\mathrm{Tr}}^m_1(ax+x^{q-2})}
=\sigma_{\frac{1}{i} \mod p}(K_q(a)).
\end{align*}
Hence, the proposition holds.
\end{proof}
\begin{remark}
From the above proposition, to consider the value $1-\frac{2}{\zeta^i+\zeta^{-i}}$  of  Kloosterman sums, we just consider the case for $K_q(a)=1-\frac{2}{\zeta+\zeta^{-1}}$.
Further, denote $N_i=\#\{a\in \mathbb{F}_q:K_q(a)=1-\frac{2}{\zeta^i+\zeta^{-i}}\}$ for $1\le i \le \frac{p-1}{2}$. From the above proposition, $N_1=N_2=\cdots=N_{\frac{p-1}{2}}$. That explains the experiment result in \cite{JZH}
that  $N_i$ are equal.
\end{remark}

\begin{theorem}\label{a^4}
Let $p$ be a prime greater than $13$, and $m\ge 1$. If $K_q(a)=1-\frac{2}{\zeta+\zeta^{-1}}$, then
\begin{align*}
\mathrm{\mathrm{Tr}}^m_1(a)=-\frac{1}{2}  ,~~\mathrm{\mathrm{Tr}}^m_1(a^2)=-\frac{1}{3} ,~~\mathrm{\mathrm{Tr}}^m_1(a^3)= -\frac{1}{5},~~\mathrm{\mathrm{Tr}}^m_1(a^4)= -\frac{136}{1155}.
\end{align*}
Further,
\begin{align*}
\sum_{i<j}a^{p^i+p^j}=\frac{7}{24}, \sum_{i<j<k}a^{p^i+p^j+p^k}=-\frac{41}{240},
\sum_{i<j<k<l}a^{p^i+p^j+p^k+p^l}=&\frac{8879}{88704}.
\end{align*}
\end{theorem}
\begin{proof}
From Proposition \ref{a8} and Proposition \ref{zetapi},
\begin{align*}
a_2=&-\mathrm{\mathrm{Tr}}^m_1(a)=\frac{1}{2} \mod p,\\
a_4=&\frac{1}{4}(\mathrm{\mathrm{Tr}}^m_1(a^2)-2(\mathrm{\mathrm{Tr}}^m_1(a))^2)=-\frac{5}{24} \mod p,\\
a_6=&-\frac{1}{36}(4\mathrm{\mathrm{Tr}}^m_1(a^3)+6(\mathrm{\mathrm{Tr}}^m_1(a))^3-9(\mathrm{\mathrm{Tr}}^m_1(a))(\mathrm{\mathrm{Tr}}^m_1(a^2)))=\frac{61}{720} \mod p;\\
a_8=&-\frac{1}{576}(24(\mathrm{\mathrm{Tr}}^m_1(a))^4-72(\mathrm{\mathrm{Tr}}^m_1(a))^2(\mathrm{\mathrm{Tr}}^m_1(a^2))+64(\mathrm{\mathrm{Tr}}^m_1(a))(\mathrm{\mathrm{Tr}}^m_1(a^3))\\
&+18(\mathrm{\mathrm{Tr}}^m_1(a^2))^2-33\mathrm{\mathrm{Tr}}^m_1(a^4))=-\frac{277}{8064} \mod p.
\end{align*}
Since $p\ge 13$, we have
\begin{align*}
\mathrm{\mathrm{Tr}}^m_1(a)=-\frac{1}{2}  ,~~\mathrm{\mathrm{Tr}}^m_1(a^2)=-\frac{1}{3} ,~~\mathrm{\mathrm{Tr}}^m_1(a^3)= -\frac{1}{5},~~\mathrm{\mathrm{Tr}}^m_1(a^4)= -\frac{136}{1155}.
\end{align*}
From Lemma \ref{Fmsum}, we have
\begin{align*}
\sum_{i<j} a^{p^i+p^j}=&\frac{1}{2}((\mathrm{\mathrm{Tr}}^m_1(a))^2-\mathrm{\mathrm{Tr}}^m_1(a^2)),\\
\sum_{i<j<k} a^{p^i+p^j+p^k}=&\frac{1}{6}((\mathrm{\mathrm{Tr}}^m_1(a))^3-3(\mathrm{\mathrm{Tr}}^m_1(a))(\mathrm{\mathrm{Tr}}^m_1(a^2))+2\mathrm{\mathrm{Tr}}^m_1(a^3)),\\
\sum_{i<j<k<l} a^{p^i+p^j+p^k+p^l}=&\frac{1}{24}((\mathrm{\mathrm{Tr}}^m_1(a))^4-6(\mathrm{\mathrm{Tr}}^m_1(a))^2
(\mathrm{\mathrm{Tr}}^m_1(a^2))+3(\mathrm{\mathrm{Tr}}^m_1(a^2))^2+8(\mathrm{\mathrm{Tr}}^m_1(a))(\mathrm{\mathrm{Tr}}^m_1(a^3))-6\mathrm{\mathrm{Tr}}^m_1(a^4)).
\end{align*}
Hence,
\begin{align*}
\sum_{i<j}a^{p^i+p^j}=\frac{7}{24},
\sum_{i<j<k}a^{p^i+p^j+p^k}=-\frac{41}{240},
\sum_{i<j<k<l}a^{p^i+p^j+p^k+p^l}=\frac{8879}{88704}.
\end{align*}
\end{proof}
\begin{remark}
When $p=7,11$, the following results also hold.
\begin{align*}
\mathrm{\mathrm{Tr}}^m_1(a)=-\frac{1}{2}  ,~~\mathrm{\mathrm{Tr}}^m_1(a^2)=-\frac{1}{3} ,~~\mathrm{\mathrm{Tr}}^m_1(a^3)= -\frac{1}{5},\\
\sum_{i<j}a^{p^i+p^j}=\frac{7}{24},
\sum_{i<j<k}a^{p^i+p^j+p^k}=-\frac{41}{240}.
\end{align*}
\end{remark}
\begin{theorem}\label{1}
Let $p\ge 7$, and $i$ be an integer. If $a\in \mathbb{F}_p$, then $K_q(a)\neq 1-\frac{2}{\zeta^i+\zeta^{-i}}$.
\end{theorem}
\begin{proof}
We just need to consider the case of $0\le i\le \frac{p-1}{2}$. If $i=0$,
this theorem holds \cite{KRV}.
From Proposition \ref{ai^2}, we just need to prove that if $a\in \mathbb{F}_p$, $K_q(a)\neq 1-\frac{2}{\zeta+\zeta^{-1}}$. Suppose that $K_q(a)= 1-\frac{2}{\zeta+\zeta^{-1}}$.
From Theorem \ref{a^4}, we have
\begin{align*}
\mathrm{\mathrm{Tr}}^m_1(a)=ma=-\frac{1}{2}  ,~~\mathrm{\mathrm{Tr}}^m_1(a^2)=ma^2=-\frac{1}{3} ,~~\mathrm{\mathrm{Tr}}^m_1(a^3)=ma^3= -\frac{1}{5}.
\end{align*}
Then
\begin{align*}
a=\frac{-\frac{1}{3}}{-\frac{1}{2}}\equiv \frac{-\frac{1}{5}}{-\frac{1}{3}} \mod p,
\end{align*}
i,e, $\frac{2}{3}\equiv \frac{3}{5} \mod p$ or $\frac{1}{15}\equiv 0 \mod p$, which is impossible. Hence, this theorem holds.
\end{proof}

\begin{theorem}\label{11}
If $p=11$, then  for any $a\in \mathbb{F}_{p^m}$ and any integer $i$, $K_{p^m}(a)\neq 1-\frac{2}{\zeta^i+\zeta^{-i}}.$
\end{theorem}
\begin{proof}
We just need to consider the case for $0\le i\le \frac{p-1}{2}$. If $i=0$,
this theorem holds \cite{KRV}. From Proposition \ref{ai^2}, we just need to prove
that if $a\in \mathbb{F}_q$ and $q=p^m$, $K_q(a)\neq 1-\frac{2}{\zeta+\zeta^{-1}}$.
From the remark after Theorem \ref{a^4}, we have
\begin{align*}
\mathrm{\mathrm{Tr}}^m_1(a)=-\frac{1}{2}  ,~~\mathrm{\mathrm{Tr}}^m_1(a^2)=-\frac{1}{3} ,~~\mathrm{\mathrm{Tr}}^m_1(a^3)= -\frac{1}{5}.
\end{align*}
Note that $33\mathrm{\mathrm{Tr}}^m_1(a)=0 \mod 11$. From Result (2) in Proposition \ref{a8}, $a_8=4\mod 11$. From Proposition \ref{zetapi}, $a_8=-2\mod 11$, which makes a contradiction. Hence, this theorem holds.
\end{proof}
Let $p\ge 13$, $s\in \{2,3,4\}$, $s|m$, $p\nmid m$, $a\in \mathbb{F}_{p^s}$, and the degree of the minimal polynomial of $a$ be $s$.
Denote $c_1=\sum_{0\le i\le s-1}a^{p^i}$, $c_2=\sum_{0\le i<j\le s-1}a^{p^i+p^j}$, $c_3=\sum_{0\le i<j<k\le s-1}a^{p^i+p^j+p^k}$, and
$c_4=\sum_{0\le i<j<k<l\le s-1}a^{p^i+p^j+p^k+p^l}$. Then the minimal polynomial of $a$ is of the form $m_a(x)=x^s-c_1 s+\cdots+(-1)^sc_s$. If $i>s$, $c_i=0$.  The coefficients
$c_i$ can be computed by $\mathrm{\mathrm{Tr}}^s(a^j)$ and the proof is similar to Lemma
\ref{Fmsum}.

\begin{theorem}\label{2}
Let $i$ be an integer, $p\ge 13$, $2\mid m$, and $a\in \mathbb{F}_{p^2}$. Then $K_q(a)\neq1-\frac{2}{\zeta^i+\zeta^{-i}}$.
\end{theorem}
\begin{proof}
We just need to prove that $K_q(a)\neq1-\frac{2}{\zeta+\zeta^{-1}}$.
If $a\in \mathbb{F}_p$, this theorem holds. Suppose that
$a\in \mathbb{F}_{p^2}\backslash \mathbb{F}_{p}$ and $K_q(a)=1-\frac{2}{\zeta+\zeta^{-1}}$. From Theorem \ref{a^4},
\begin{align*}
\mathrm{\mathrm{Tr}}^m_1(a)&=\frac{m}{2}\mathrm{\mathrm{Tr}}^2_1(a)=-\frac{1}{2},\mathrm{\mathrm{Tr}}^m_1(a^2)=
\frac{m}{2}\mathrm{\mathrm{Tr}}^2_1(a^2)=-\frac{1}{3},\\
\mathrm{\mathrm{Tr}}^m_1(a^3)&=\frac{m}{2}\mathrm{\mathrm{Tr}}^2_1(a^3)=-\frac{1}{5},
\mathrm{\mathrm{Tr}}^m_1(a^4)=\frac{m}{2}\mathrm{\mathrm{Tr}}^2_1(a^4)=-\frac{136}{1155}.
\end{align*}
If $p| m$, $\frac{1}{2}\equiv 0 \mod p$, which is impossible. Hence, this theorem holds.

Let $p\nmid m$. We have
\begin{align*}
\mathrm{\mathrm{Tr}}^2_1(a)&=-\frac{1}{2}r,\mathrm{\mathrm{Tr}}^2_1(a^2)=-\frac{1}{3}r,\\
\mathrm{\mathrm{Tr}}^2_1(a^3)&=-\frac{1}{5}r,\mathrm{\mathrm{Tr}}^2_1(a^4)=-\frac{136}{1155}r,
\end{align*}
where $r\equiv \frac{2}{m}\mod p$. We can compute
\begin{align*}
c_1=&-1/2r,c_2=1/8r^2 + 1/6r,\\
c_3=&-1/48r^3 - 1/12r^2 - 1/15r,c_4=1/384r^4 + 1/48r^3 + 17/360r^2 + 34/1155r,
\end{align*}
Note that $c_3=c_4=0$ and $r\neq 0 \mod p$, then $r=-108/77$. Hence,
\begin{align*}
c_1=54/77,c_2=72/5929,c_3=-29556/2282665,c_4=147780/35153041.
\end{align*}
Since $c_3=c_4\equiv 0\mod p$, $p|\mathrm{gcd}(29556,147780)=2^2\cdot3^2\cdot821$.
Hence, if $p\ge13$ and $p\neq 821$, this theorem holds.

If $p=821$, then  $c_1=86$, $c_2=659$, and $m_a(x)=x^2-c_1x+c_2=(x-300)(x-607)$,
which contradicts that $m_a(x)$  is the minimal polynomial of $a$.

Hence, this theorem holds.
\end{proof}

\section{Nonexistence of some binomial regular bent functions}

\begin{theorem}
Let $p\ge 7$, $n=2m$. Let  $a\in \mathbb{F}_{p^n}$, and  $a^{p^m+1},b\in \mathbb{F}_p$. Let
$t$ be an integer  satisfying $\mathrm{gcd}(t,p^m+1)=1$. Then the following $p-$ary function
$$f_{a,b,t}(x)=\mathrm{\mathrm{Tr}}^n_1(ax^{t(p^m-1)})+bx^{\frac{p^n-1}{2}}$$
can not be a regular  bent function.
\end{theorem}
\begin{proof}
From Theorem \ref{Zh} and Theorem \ref{1}, this theorem can be obviously obtained.
\end{proof}
\begin{theorem}
Let $p\ge 13$, $n=2m$, $2|m$. Let  $a\in \mathbb{F}_{p^n}$, $a^{p^m+1}\in \mathbb{F}_{p^2}$, and $b\in \mathbb{F}_{p}$. Let $t$ be an integer
 satisfying $\mathrm{gcd}(t,p^m+1)=1$.
Then the $p-$ary function
$$f_{a,b,t}(x)=\mathrm{\mathrm{Tr}}^n_1(ax^{t(p^m-1)})+bx^{\frac{p^n-1}{2}}$$
can not be a regular bent  function.
\end{theorem}
\begin{proof}
From Theorem \ref{Zh} and Theorem \ref{2}, this theorem can be obviously obtained.
\end{proof}
\begin{theorem}
Let $n=2m$, $2|m$. Let $a\in \mathbb{F}_{11^n}$, and $b\in \mathbb{F}_{11}$.
Let $t$  be an integer satisfying $\mathrm{gcd}(t,11^m+1)=1$. Then the $p-$ary function
$$f_{a,b,t}(x)=\mathrm{\mathrm{Tr}}^n_1(ax^{t(11^m-1)})+bx^{\frac{11^n-1}{2}}$$
can not be a regular bent function.
\end{theorem}
\begin{proof}
From Theorem \ref{Zh} and Theorem \ref{11}, this theorem can be obviously obtained.
\end{proof}
\section{Conclusion}
This paper discusses the special value $1-\frac{2}{\zeta+\zeta^{-1}}$ of Kloosterman
sum $K_q(a)$($q=p^m$), and presents necessary conditions for $a$ such that $K_q(a)=1-\frac{2}{\zeta+\zeta^{-1}}$. We prove that for $p=11$, there does not
exist $a$ satisfying $K_q(a)=1-\frac{2}{\zeta+\zeta^{-1}}$. For general
$p\ge 13$, we prove that if $a\in \mathbb{F}_{p^s}$ and $s=\mathrm{gcd}(2,m)$,
there does not exist $a$ satisfying $K_q(a)\neq 1-\frac{2}{\zeta+\zeta^{-1}}$.
From results in this paper, it seems that for $p\ge 7$ there does not exist
$a$ satisfying $K_q(a)=1-\frac{2}{\zeta+\zeta^{-1}}$. Our further work will consider
the generalization of our techniques to general cases.

\section*{Acknowledgment}
This work was supported by
the Natural Science Foundation of China
(Grant No.10990011 \& No. 61272499). Yanfeng Qi also acknowledges support from Aisino Corporation Inc.


\ifCLASSOPTIONcaptionsoff
  \newpage
\fi


\begin{thebibliography}{99}

\bibitem{CHZ}Charpin, Pascale, Tor Helleseth, and Victor Zinoviev. "Divisibility properties of classical binary Kloosterman sums." Discrete Mathematics 309, no. 12 (2009): 3975-3984.
\bibitem{GK}   Gross, Benedict H., and Neal Koblitz. "Gauss sums and the p-adic G-function." The Annals of Mathematics 109, no. 3 (1979): 569-581.
\bibitem{GF}Gologlu, Faruk. "Ternary Kloosterman sums modulo 4." Finite Fields and Their Applications 18, no. 1 (2012): 160-166.
\bibitem{GFMM}Gologlu, Faruk, Gary McGuire, and Richard Moloney. "Ternary Kloosterman sums modulo 18 using Stickelberger's theorem," In Sequences and Their Applications,SETA 2010, pp. 196-203. Springer Berlin Heidelberg, 2010.
\bibitem{GFLGM}Gologlu, Faruk, Petr Lisonek, Gary McGuire, and Richard Moloney. "Binary Kloosterman sums modulo 256 and coefficients of the characteristic polynomial." Information Theory, IEEE Transactions on 58, no. 4 (2012): 2516-2523.
\bibitem{HK}T. Helleseth and A. Kholosha, ''monomial and quadratic bent functions
over the finite fields of odd characteristic,''IEEE Trans. Inf. Theory , vol.
52, no. 5, pp. 2018-2032, May 2006.
\bibitem{JZH}Jia, Wenjie, Xiangyong Zeng, Tor Helleseth, and Chunlei Li. "A class of binomial bent functions over the finite fields of odd characteristic." Information Theory, IEEE Transactions on 58, no. 9 (2012): 6054-6063.
\bibitem{KRV}K. P. Kononen, M. J. Rinta-aho, and K. O. V\"{a}\"{a}n\"{a}nen, ''on integer
values of Kloosterman sums,''IEEE Trans. Inf. Theory, vol. 56, no.
8, pp. 4011-4013, Aug. 2010.
\bibitem{Ld}  G. Leander, ¡°Monomial bent functions,¡± IEEE Trans. Inf. Theory, vol.
2, no. 52, pp. 738¨C743, 2006.
\bibitem{Ms24}Moisio, Marko. "The divisibility modulo 24 of Kloosterman sums on ~$GF(m^2)$,m even." Finite Fields and Their Applications 15, no. 2 (2009): 174-184.
\bibitem{Ms}M. Moisio, ''on certain values of Kloosterman sums,''IEEE Trans. Inf.
Theory , vol. 55, no. 8, pp. 3563-3564, Aug. 2009.
\bibitem{Ml} R. Moloney,Divisibility properties of Kloosterman sums and division polynomials for Edwards curves[D]. University College Dublin, 2011.
\bibitem{S}Serre, Jean-Pierre. Local fields. Vol. 67. New York: Springer-Verlag, 1979.
\bibitem{W}Wan, Daqing. "Minimal polynomials and distinctness of Kloosterman sums." Finite Fields and Their Applications 1, no. 2 (1995): 189-203.
\bibitem{ZYH}Dabin Zheng, Long Yu, Lei Hu,"On a class of binomial bent functions over the finite fields of odd characteristic,"Applicable Algebra in Engineering, Communication and Computing,Volume 24, Issue 6 , pp 461-475.


\end{thebibliography}
\end{document}